\documentclass[runningheads]{llncs}

\usepackage{amsmath, amssymb, amsfonts}
\usepackage{mathtools}
\usepackage{graphicx}
\usepackage{booktabs}
\usepackage{microtype}
\usepackage{enumitem}
\usepackage{kotex}
\usepackage{comment}
\usepackage{hyperref}
\hypersetup{colorlinks=true, linkcolor=blue, citecolor=blue, urlcolor=blue}
\usepackage{tabularx}
\usepackage{adjustbox}

\newtheorem{assumption}{Assumption}
\newtheorem{observation}{Observation}

\begin{document}

\title{(Im)possibility of Incentive Design for Challenge-based Blockchain Protocols}

\author{Suhyeon Lee\inst{1,2} \and Dieu-Huyen Nguyen\inst{1} \and Donghwan Lee\inst{1}}
\institute{Tokamak Network \and Hashed Open Research}

\maketitle

\begin{abstract}
Blockchains offer a decentralized and secure execution environment strong enough to host cryptocurrencies, but the state-replication model makes on-chain computation expensive. To avoid heavy on-chain workloads, systems like Truebit and optimistic rollups use challenge-based protocols, performing computations off-chain and invoking the chain only when challenged. This keeps normal-case costs low and, if at least one honest challenger exists, can catch fraud. What has been less clear is whether honest challengers are actually incentivized and a dishonest proposer is properly damaged under the worst case environment. We build a model with a colluding minority, heterogeneous costs, and three ordering modes. We then ask whether two goals can be met together: honest non-loss and fraud deterrence. Our results are clear: in single-winner designs, the incentive design is impossible or limited in scale. By contrast, in multi-winner designs, we obtain simple, explicit conditions under which both goals hold.

\end{abstract}

\section{Introduction}

Blockchains offer a decentralized, verifiable execution environment. Because many nodes must replicate state, on-chain computation is costly. To scale, systems shift intensive work off-chain and verify correctness only when disputes arise through challenge protocols, including Truebit-style verification, fraud proofs in optimistic rollups, and data-availability audits.
These mechanisms keep the common case cheap while retaining a safety valve: if at least one honest party challenges, fraud can be surfaced and a slashed deposit funds the payout. It seems very easy to satisfy.
However, this impression is not enough to construct a blockchain protocol for users considering its nature as \textit{verifier's dilemma} \cite{luu2015demystifying}. Also, even if an honest verifier is not lazy, it does not guarantee that the verification will be compensated enough. For example, if a dishonest proposer is disputed by a challenger account which is actually the proposer's Sybil account, the dishonest proposer will be able to avoid any loss. In this paper, we raise two questions. \emph{First, are honest challengers incentivized in the blockchain environments?} The first question is to guarantee at least one honest validator exists. \emph{Second, do dishonest proposers who post wrong values get proper economic damage that the system intends?} The second question is to guarantee dishonest proposers do not appear.

To answer the question, we develop a mechanism-agnostic model that captures different ordering properties and asks whether two goals can be met together: honest non-loss for included challengers and explicit deterrence against the proposer. The model features multiple challengers (validators) with heterogeneous but bounded costs, a proposer deposit and reward share, and a colluding subset of size $A<N/2$. Ordering is treated in three modes that differ in who captures priority value: fair ordering (random within short arrival intervals), unfair-by-builder, and unfair-by-proposer. We analyze single-winner and multi-winner with unbounded inclusion. Our contributions are as follows:
\begin{itemize}
\item \textbf{Problem formalization.} Prior work mainly posed the challenge via examples; we provide a clean formal model that makes the incentive question analyzable.
\item \textbf{Single-winner limits.} We show that, under realistic ordering assumptions, single-winner challenge designs are either impossible to incentivize or tightly constrained.
\item \textbf{Multi-winner feasibility.} We model multi-winner designs and prove that, unlike single-winner schemes, these systems generally make it easier to design fair and effective incentives.
\end{itemize}

\noindent The rest of the paper is organized as follows.
Section~\ref{sec:related} briefly reviews related work.
Section~\ref{sec:model} formalizes the model and design goals.
Section~\ref{sec:single} presents single-winner limits under proposer- and builder-ordered priority and the scalability bound under fair ordering.
Section~\ref{sec:multi} establishes feasibility for multi-winner and shows scalability-free calibration.
Section~\ref{sec:discussion} discusses the design trade-offs and implications.
Section~\ref{sec:conclusion} concludes with implications and implementation guidance.

\section{Related Work}
\label{sec:related}

Prior work develops concrete dispute protocols and fraud-proof algorithms for rollups, including BoLD and Dave \cite{alvarez2024bold,nehab2024dave}. These protocols, however, enforce single-winner designs through duplicate prevention mechanisms at the protocol level. Separate strands analyze validator behavior, covering incentive alignment \cite{mamageishvili2023incentive} and diligence under economic frictions \cite{AFT2024diligence}, and study censorship risks in fraud-proof timelines \cite{berger2025economic}. Announcement-game dynamics \cite{AAMAS2025accouncement} and risks around sequencer centralization and censorship \cite{motepalli2023sok} are documented alongside broader rollup attack surfaces \cite{koegl2023attacks} and accountable designs \cite{tas2022accountable}. Within this space, Lee \cite{lee2025wtsc} explicitly articulates a vulnerability of fraud-proof incentive systems in optimistic-rollup dispute games and sketches countermeasures, but does not provide a general feasibility criterion or proofs that rule out the failure modes. Our work provides that missing layer via a compact formal model that separates ordered priority and single- versus multi-winner settings, proving single-winner limits and a feasibility boundary for multi-winner inclusion.

\section{Model and Design Goals}
\label{sec:model}

In this section, we formalize a minimal model for challenge-based rewards: players, costs with an upper bound, deposits, ordering modes, and winner multiplicity. This paper follows a \textit{mechanism design} approach, which can be conceptualized as the inverse of traditional game theory analysis, in which the objective is designing game rules to achieve desired outcomes instead of analyzing all possible outcomes from existing, fixed rules.
We state two design goals used in the following sections.

\subsection{Players, costs, and deposits}
The game is modeled in normal form under complete information, such that the entire payoff structure is common knowledge among all players, who then choose their strategies simultaneously. We consider $N$ challengers to a specific proposal and assume a conservative Sybil and collusion bound where only a subset of challengers with $A=|\mathcal{A}| < N/2$ may collude with the proposer $P$. Pooling mechanisms such as BoLD's AssertionStakingPool count as one challenger, since the protocol registers only the pool address as a single economic entity \cite{arbitrum_bold_2024}. Each challenger stakes the same deposit $s$, and to submit a successful valid challenge they incur a cost $c_i>0$. The cost $c_i$ comprises of the initial cost to identify fraud and the cost to process the challenge till the end (e.g. executing the smart contract to prove).Costs may vary across agents and across challenge processes, but admit known finite upper bounds. Let
\[
c_i \;=\; c_{i}^{\textsf{init}} + c_{i}^{\textsf{proc}}, \qquad
\tilde c^{\textsf{init}} \;\ge\; \sup_i c_{i}^{\textsf{init}}, \qquad
\tilde c^{\textsf{proc}} \;\ge\; \sup_i c_{i}^{\textsf{proc}} .
\]
Define the aggregate worst-case cost bound as
\[
\tilde c \;:=\; \tilde c^{\textsf{init}} + \tilde c^{\textsf{proc}},
\]
so that $c_i \le \tilde c$ for all $i$.
The bound $\tilde c$ serves as a worst-case guarantee for verifying individual rationality (IR).

\subsection{Reward pool, winners, and splits}

\paragraph{Slashing.}
Upon acceptance of a valid fraud proof, the proposer's deposit $D_p$ is slashed.
A fraction $\alpha D_p$ is paid to challengers designated as winners and the remainder $(1-\alpha)D_p$ is burned, with $\alpha\in(0,1]$. Even though burning a part of the slashed deposit in challenge-based protocols is not generally adopted, we introduce this rule to force a minimal economic damage under a possibility of collusion.

\paragraph{Winner set.}
The system has the number of winners $m$. The single-winner setting $m=1$ is a special case which does not allow multiple winners.
In the multi-winner setting, the protocol imposes no ex-ante numeric cap on winners. However, runtime constraints
(for example, dispute-window length and blockchain throughput) imply that the realized number of winners $m$ is finite and ex-post bounded.

\paragraph{Split rule.}
We adopt equal split among winners: each of the $m$ included winners receives $\alpha D_p/m$. That is, with fewer winners, each winner receives a larger share. \\

Lastly, we add an assumption to consider the worst reward of the honest challengers and an observation follows.

\begin{assumption}[Best-effort participation] \label{assumption: best-effort}

Validators try to submit a valid challenge during the dispute period whenever a public conservative bound on net payoff is non-negative using an exogenous upper bound on fees and the announced inclusion capacity. The realized winners do not exceed capacity or the number of valid challengers, and in the worst case the capacity is saturated.
\end{assumption}

Colluding challengers aligned with the fraudulent proposer would otherwise have no need to challenge, but since honest challengers engage under this assumption, they must also submit valid challenges to recapture any share of the reward.

\begin{observation}[Coalition recapture]\label{obs:recapture}
If at least one honest challenger expects a nonnegative payoff from submitting a valid challenge,
then the proposer-validator coalition has a best response to also submit a valid challenge to recapture reward that comes from the slashed deposit.
This strategy is strictly profitable when sequencing fees are internally rebated (e.g., proposer-run auction),
and non-worse under competitive external auctions.
\end{observation}

\subsection{Ordering and fees: fair vs.\ unfair}
\label{subsec:ordering}

After validity is checked, an ordering mechanism fixes a realized order $\sigma$ of valid challenges, and inclusion takes the first $m$ positions designated as winners (the case $m=N$ corresponds to systems that include all valid challenges. 
The mechanism’s impact depends critically on how this ordering is performed.
We define fair ordering as content-agnostic transaction sequencing. Its opposite, unfair ordering, sets order on economic grounds. Today, the majority of public chains and applications effectively operate under the unfair ordering.

\textbf{Unfair ordering (\textsf{U}).} We categorize it into two different cases. Especially, the second case implies that the proposer can order the winners' sequence regardless of the challenge submission order.

\begin{enumerate}

\item \textbf{(\textsf{U-B}) Builder-ordered priority:}
An upstream L1/L2 builder (or block producer) \emph{sells sequencing priority via an auction}.
Let the eventual winner be $w$ with private value $v_w := \alpha D_p/m - c_w^{\textsf{proc}}$. \footnote{Once validity costs are sunk, bidders value a priority slot by $v_i=\alpha D_p/m - c_i^{\textsf{proc}}$, and competition drives winner surplus near zero.}
Under market efficiency,
the clearing fee satisfies
$ f^{\star} \;\le\; v_w, \text{ for any small bid increment } \varepsilon>0. $
The fee is paid to the builder (not to $P$) and should be understood as inclusive of realistic ordering competition.

\item \textbf{(\textsf{U-P}) Proposer-ordered priority:}
Priority is sold via an on-chain auction\footnote{For example, implementations that guarantee at least fair competition as discussed in Lee \cite{lee2025wtsc}. We do not consider implementations where $P$ can arbitrarily sequence all dispute games within a chess clock without paying any fee, since that would deny honest challengers any meaningful opportunity.} implemented by the proposer $P$. 
Under the same market-clearing logic as in \textsf{U-B}, honest challengers who seek priority pay $f^{\star}$ (to $P$). For colluding or Sybil challengers, this payment is internally recycled within the coalition, so their \emph{effective} priority cost in payoffs is $f=0$.\footnote{The UP model illustrates a ``worst-case scenario'' by assuming a malicious proposer has total control over transaction ordering. Although the challenge period usually covers several blocks, this model specifically highlights the maximum advantage an attacker could exploit within a single block.}

\end{enumerate}

\textbf{Fair ordering (\textsf{F}).} Valid challenges that arrive within a short time window $\Delta$ (a simultaneous-arrival interval) are ordered by a symmetric, identity- and fee-agnostic random permutation (random tie-breaking). There is no priority fee ($f=0$) in this regime.

\subsection{Key Parameters and Payoff Matrix}

We consider a game between an Honest Challenger (HC) and the Adversarial Coalition (AC). The HC always challenges a potentially fraudulent transaction, while the AC decides whether to Commit Fraud. We assume that once an honest challenger initiates a dispute, the fraud is always detected. The loss of the adversarial coalition is determined as $D_p - \phi\alpha D_p$ = $(1-\phi\alpha) D_p$. This loss is proportional to the size of coalition $\phi$ and the reward fraction $\alpha$ of proposer deposit  if fraud proof success. A summary of notations is provided in Table~\ref{tab:parameters}, and the detailed payoff matrix for the single-winner game is presented in Table~\ref{tab:payoff_matrix}.

\begin{table}[hb]
\centering
\caption{Description of Parameters}
\label{tab:parameters}
\renewcommand{\arraystretch}{1.2} 

\begin{adjustbox}{width=0.92\columnwidth, center}
\begin{tabular}{@{}l @{\hspace{0.2cm}} l p{8cm}@{}} 
\toprule
\textbf{Symbol} & \textbf{Parameter} & \textbf{Description} \\
\midrule
$c_i^{init}$ & Challenger Initial Cost & The initial cost for an honest challenger to verify and submit a challenge. \\
$c_i^{proc}$ & Challenger Processing Cost & The processing cost for an honest challenger. \\
$c_i$ & Total Challenger Cost & The total cost for an honest challenger, including initial and processing costs. \\
$D_p$ & Proposer's Deposit & The deposit made by the proposer. \\
$\alpha$ & Reward Fraction & The fraction of the proposer's deposit ($D_p$) distributed as a reward if fraud proof succeeds. \\
$\phi$ & Colluding Fraction & The fraction of colluding challengers within the adversarial coalition (equivalent to $A/N$). \\
$\eta$ & Min. Deterrence Level & The minimum deterrence level required to prevent fraud. \\
$f$ & Priority Fee & The priority fee a validator pays in a \textsf{U-b} or \textsf{U-P} setup. \\
$f^*$ & Market-Clearing Fee & The market-clearing priority fee. \\
\bottomrule
\end{tabular}
\end{adjustbox}
\end{table}

\begin{table}[tb]
\centering
\caption{Payoff Matrix for Single-Winner Game}
\begin{tabular*}{\textwidth}{@{\extracolsep{\fill}}lcc}
\toprule
\textbf{Challenger Action} & \textbf{AC: Commit Fraud} & \textbf{AC: Not Commit Fraud} \\
\midrule
HC: Challenge (\textsf{U-P}) & $\alpha D_p - c_i - f^*; -(1-\phi\alpha) D_p$ & $- c_i - f^*; 0$ \\
HC: Challenge (\textsf{U-b}) & $\alpha D_p - c_i - f^*; -(1-\phi\alpha) D_p$ & $- c_i - f^*; 0$ \\
HC: Challenge (Fair) & $\alpha D_p - c_i; -(1-\phi\alpha) D_p$ & $- c_i; 0$ \\
\bottomrule
\end{tabular*} %
\label{tab:payoff_matrix}
\end{table}

Under the multi-winner setting, the economic value of transaction ordering is significantly diminished. Therefore, the ordering fee is nearly zero with all the order fairness environments (\textsf{U-P}, \textsf{U-b}, and \textsf{F}). Each of the $N$ included winners is eligible for $\frac{1}{N}$ of the reward. Table~\ref{tab:payoff_matrix_multi} illustrates the payoff matrix under this multi-winner setting.

\begin{table}[h!]
\centering
\caption{Payoff Matrix for Multi-Winner Game}
\begin{tabular*}{\textwidth}{@{\extracolsep{\fill}}lcc}
\toprule
\textbf{Challenger Action} & \textbf{AC: Commit Fraud} & \textbf{AC: Not Commit Fraud} \\
\midrule
HC: Challenge & $\frac{1}{N}\alpha D_p - c_i; -(1-\phi\alpha) D_p$ & $- c_i; 0$ \\
\bottomrule
\end{tabular*}
\label{tab:payoff_matrix_multi}
\end{table}

\subsection{Design goals}
We require two design goals for the protocol. One is to economically motivate an honest challenger for security. The other is to demotivate malicious block proposer for security.

\textbf{\textsf{O1} (Ex-ante IR at participation stage)}\label{def:O1-exante}
In the worst-case scenario where collusion is allowed and the proposer submits a fraudulent state root, every honest challenger
$i\notin\mathcal{A}$ must have non-negative \emph{ex-ante} expected utility before the initial participation decision:
\[
\mathrm{E}[U_i]\;\ge\;0.
\]

\textbf{\textsf{O2} ($\eta$-fraction deterrence).}
We explicitly require that the adversarial coalition (proposer plus colluding challengers) incurs a loss of at least an $\eta$-fraction of the deposit whenever fraud is detected:
\[
\mathrm{Loss}_{\mathrm{adv}} \;\ge\; \eta D_p
\qquad (\eta\in(0,1)).
\]

These goals yield the following two lemmas and one corollary.

\begin{lemma}[Worst-case lower bound on $\alpha$ for the reward goal \textsf{O1}]
\label{lem:alpha-lower-bound}
To hold \textsf{O1}, it is necessary and sufficient that
\[
\alpha \;\ge\; \alpha_L(m,f) \;:=\; \frac{m\,(\tilde c + f)}{D_p}.
\]
\end{lemma}
\begin{proof}
An included honest winner earns $\alpha D_p/m - c_i - f$. The worst case is $c_i=\tilde c$, so \textsf{O1} is equivalent to $\alpha D_p/m - \tilde c - f \ge 0$, i.e., $\alpha \ge m(\tilde c+f)/D_p$. The converse is immediate.
\end{proof}

\begin{lemma}[Upper bound on $\alpha$ for the deterrence goal \textsf{O2}]
\label{lem:alpha-upper-bound}
For $\alpha\in(0,1]$ and $\phi=\frac{A}{N}\in[0,1]$, \textsf{O2} implies
\[
\alpha \;\le\; \min\!\left\{1,\frac{1-\eta}{\phi}\right\},
\]
with the convention that when $\phi=0$ the bound reduces to $\alpha\le 1$.
\end{lemma}
\begin{proof}
From \textsf{O2},
\(
D_p - \phi\,\alpha D_p \ge \eta D_p
\)
iff
\(
\phi\,\alpha \le 1-\eta
\).
If $\phi>0$, then \(\alpha \le (1-\eta)/\phi\); in all cases \(\alpha\le 1\).
\end{proof}

\begin{corollary}[$\phi$-free nontrivial lower bound on $\eta$ for the deterrence goal \textsf{O2}]
\label{cor:eta-lower-bound}
Under Assumption~\ref{assumption: best-effort} (so $\phi<\tfrac12$), enforcing a $\phi$-free conservative deterrence bound via
\[
\alpha \;\le\; 2(1-\eta)
\]
is nontrivial (i.e., $<1$) if and only if $\eta>\tfrac12$.
\end{corollary}
\begin{proof}
Replacing $\phi$ in Lemma~\ref{lem:alpha-upper-bound} by its worst admissible value $\phi_{\max}=\tfrac12$ yields the φ-free sufficient bound $\alpha\le 2(1-\eta)$. This is $<1$ iff $2(1-\eta)<1$, i.e., $\eta>\tfrac12$.
\end{proof}

\subsection{Model Scope and Limitations}
Our model focuses on a single-stage challenge game. This is a simplification of real-world protocols like BoLD, Dave, and Cannon, which involve multiple rounds of proposals and challenges. This single-stage abstraction allows us to isolate and formally analyze the fundamental incentive problem: how to ensure that honest parties are motivated to challenge fraudulent proposals while deterring malicious proposers. The insights from this analysis provide a foundation for understanding the incentive structure of more complex multi-stage protocols.

While our model does not capture the full complexity of multi-stage protocols, the core incentive problems we identify—priority capture, auction dynamics, and reward sharing—persist across every stage. Consequently, our findings offer a framework for addressing the incentive hurdles inherent in each level of challenge-based protocols.

\section{Single-Winner: Impossibility and Scalability Limits}
\label{sec:single}

This section analyzes single-winner mechanisms ($m=1$) under three ordering regimes: proposer-ordered priority (\textsf{U-P}), builder-sold priority (\textsf{U-B}), and fair ordering (\textsf{F}). We, then, obtain results that imply impossibility or scale-limited feasibility under the settings.

\subsection{\textsf{U-P} with single-winner: \textsf{O1} fails}
\label{subsec:single-UP}

In \textsf{U-P}, the proposer $P$ sells priority directly. For colluding challengers, the effective fee is $f=0$ (coalition's internal recycling), while honest challengers have no access to internal priority.

\begin{theorem}[\textsf{O1} failure under \textsf{U-P} with a single winner]\label{thm:UP-m1-O1}
Let $q_i \in [0, 1]$ be the probability that challenger $i$ is selected as the winner. 
Under \textsf{U-P} and a single winner setting ($m=1$),
for any parameters $(\alpha,D_p)$ and any cost profile $\{c_i^{\textsf{init}},c_i^{\textsf{proc}}\}_{i=1}^N$ with $c_i^{\textsf{init}}>0$,
there exists a coalition strategy $s^\star$ such that for every honest $i\notin\mathcal{A}$,
\[
q_i=0 \quad\text{and}\quad \mathrm{E}[U_i]
= q_i(\alpha D_p - c_i^{\textsf{proc}}) - c_i^{\textsf{init}} 
= -\,c_i^{\textsf{init}}  < 0.
\]
\end{theorem}
\begin{proof}
Assume a fraudulent state. Under \textsf{U-P} with $m=1$, the unique sequencing slot is sold via an on-chain auction. Let $f^{\star}$ be the market-clearing price. 
Define a coalition strategy $s^\star$: pick any $a\in\mathcal{A}$ and bid $b_a=f^{\star}+\varepsilon$ with $\varepsilon>0$. 
The payment is transferred to $P$ and recycled within $(P,\mathcal{A})$, so the coalition's effective cost is $0$, hence $s^\star$ is feasible for any bid increment $\varepsilon>0$.

In a competitive single-slot auction, any honest bidder's equilibrium bid is $\le f^{\star}$; therefore $b_a$ wins deterministically, and the sole winner is $a$. 
Thus for every honest $i\notin\mathcal{A}$, $q_i=0$. Transfers realize only upon inclusion, so
$\mathrm{E}[U_i]=-c_i^{\textsf{init}}$ (since $c_i^{\textsf{init}}>0$), and \textsf{O1} fails.
\end{proof}

\subsection{\textsf{U-B} with single-winner: \textsf{O1} is not guaranteed}
\label{subsec:single-UB}

We now consider the builder-run auction regime (\textsf{U-B}), where challengers compete for a single winner slot by bidding to an external builder. This analysis models the interaction as a two-stage game. First, challengers incur the upfront cost $c^{\textsf{init}}$ to participate. Second, they bid for the slot.

At the bidding stage, $c^{\textsf{init}}$ is a sunk cost and thus does not influence a rational bidder's valuation. The private value $v_i$ for the slot is therefore the net payoff upon winning: $v_i := \alpha D_p - c^{\textsf{proc}}$. In a competitive auction, the winning bid $f^\star$ will approach $v_i$, driving the winner's surplus $(v_i - f^
\star)$ towards zero. This dynamic ensures that the winner's net gain is insufficient to cover the initial participation cost.

\begin{theorem}[\textsf{O1} cannot be guaranteed under \textsf{U-B} with a single winner]\label{thm:UB-m1-O1}
In the \textsf{U-B} regime with a single-winner, for any protocol parameters and any cost profile where $c_i^{\textsf{init}}>0$, an honest challenger's ex-ante expected utility is strictly negative. 
\end{theorem}

\begin{proof}
Write the honest $i$'s ex-ante payoff as
\[
\mathrm{E}[U_i] \;=\; q_i\,(v_i - f^\star) - c_i^{\textsf{init}},
\]
where $q_i\in[0,1]$ is the inclusion-and-payment probability under \textsf{U-B}.
Given $\varepsilon>0$ choose competitive conditions such that $v_i - f^\star \le \varepsilon$.
Then $\mathrm{E}[U_i] \le q_i\,\varepsilon - c_i^{\textsf{init}} < 0$ since $c_i^{\textsf{init}}>0$.
This bound is independent of $(\alpha,D_p)$, so ex-ante IR (\textsf{O1}) cannot be guaranteed.
\end{proof}

\subsection{\textsf{F} with single-winner: Fair ordering makes a scalability limitation}
\label{subsec:single-fair}

Under \textsf{F}, ties among simultaneously valid challenges are broken uniformly at random and fees are negligible ($f\simeq 0$). So, utilities depend only on rewards and costs.

\begin{theorem}[Scalability bound under fair single-winner setting]
\label{thm:single-F}
If $N$ challengers simultaneously validate and compete under \textsf{F} with $m=1$,
then ex-ante nonnegative payoff for each participant requires
\[
\alpha D_p \;\ge\; N\,\tilde c^{\textsf{init}} \;+\; \tilde c^{\textsf{proc}}.
\]
For large $N$, this conflicts with $\alpha\le \min\{1,(1-\eta)/\phi\}$, yielding an empty feasibility region.
\end{theorem}

\begin{proof}
Each honest challenger wins with probability $1/N$.
An included winner receives $\alpha D_p$ and incurs $c_i^{\textsf{proc}}$, while every participant pays $c_i^{\textsf{init}}$ upfront.
Thus $
\mathrm{E}[U_i] \;=\; \frac{1}{N}\big(\alpha D_p - c_i^{\textsf{proc}}\big) \;-\; c_i^{\textsf{init}}$.
Under worst-case costs $c_i^{\textsf{init}}=\tilde c^{\textsf{init}}$ and $c_i^{\textsf{proc}}=\tilde c^{\textsf{proc}}$,
ex-ante nonnegativity requires $\alpha D_p \ge N\,\tilde c^{\textsf{init}}+\tilde c^{\textsf{proc}}$.
Combining with \textsf{O2}, $\alpha\le \min\{1,(1-\eta)/\phi\}$, shows infeasibility for sufficiently large $N$.
\end{proof}

\section{Multi-Winner: Feasibility with Calibration}
\label{sec:multi}

In this section, we analyze the multi-winner, non-exclusion setting and show that
both goals: honest non-loss (O1) and deterrence (O2) can be achieved simultaneously without the scalability issue exists in the single-winner setting.

\subsection{Feasible Interval under Non-Exclusion}\label{subsec:feasible-interval}

Under the multi-winner setting, the economic value of transaction ordering is significantly diminished. Therefore, the ordering fee is nearly zero with all the order fairness environments (\textsf{U-P}, \textsf{U-B}, and \textsf{F}). When $f \simeq 0$, both \textsf{O1} and \textsf{O2} reduce to bounds on $\alpha$.

\begin{theorem}[Feasible interval under non-exclusion]
\label{thm:feasible-interval}
Under $m = N$ and $f = 0$, the feasible interval for $\alpha$ is
\[
\alpha \in \left[ \frac{N \tilde{c}}{D_p}, \;\min\left\{ 1, \frac{1-\eta}{\phi} \right\} \right].
\]
And, the interval is non-empty if and only if
\[
\frac{N \tilde{c}}{D_p} \;\leq\; \min\left\{ 1, \frac{1-\eta}{\phi} \right\}.
\]
\end{theorem}

\begin{proof}
\emph{\textsf{O1} lower bound.} Under non-exclusion with equal split, each of the $N$ included winners receives $\alpha D_p/N$ and incurs cost $c_i$. Using the worst-case bound $c_i \le \tilde{c}$ and $f=0$, \textsf{O1} requires
$
\frac{\alpha D_p}{N} - \tilde{c} \;\ge\; 0
\quad\Longleftrightarrow\quad
\alpha \;\ge\; \frac{N \tilde{c}}{D_p}.
$

\emph{\textsf{O2} upper bound.} Let $\phi \in [0,1]$ denote the coalition’s recapture fraction of the challengers’ total payout $\alpha D_p$. Then the adversary’s net loss ($\text{Loss}_{\text{adv}}$) is
$(1-\phi \alpha)D_p.$
\textsf{O2} requires $(1-\phi \alpha)D_p \ge \eta D_p$, i.e., $
\phi \alpha \;\le\; 1-\eta
\quad\Longleftrightarrow\quad
\alpha \;\le\; \frac{1-\eta}{\phi},
$
together with the intrinsic bound $\alpha \le 1$. 

\end{proof}

\subsubsection{Calibration Regimes}\label{subsubsec:calibration}

Depending on the upper bound, two regimes arise: Case 1: $\alpha=1$ and Case 2: $\alpha = (1-\eta)/\phi < 1$. From Case 1, we get the similar scalability-limited result in Theorem~\ref{thm:single-F}. i.e., $D_P \ge N \tilde c$. On the other hand, Case 2 brings up a scale-free feasibility as the following lemma.

\begin{lemma}[Scale-free non-emptiness]\label{lem:scale-free}
In the non-exclusion setting, the feasible interval in Theorem~\ref{thm:feasible-interval} is non-empty for all $N$ whenever
\[
D_p \;\ge\; \frac{\tilde{c}\,A}{1-\eta}.
\]
\end{lemma}

\begin{proof}
With upper bound $(1-\eta)/\phi$ and $\phi = A/N$, the non-emptiness condition from
Theorem~\ref{thm:feasible-interval} becomes
\[
\frac{N \tilde{c}}{D_p} \;\le\; \frac{1-\eta}{\phi}
\;=\; \frac{1-\eta}{A/N}
\;=\; \frac{(1-\eta)N}{A}.
\]
Cancel $N$ to obtain $\tilde{c}/D_p \le (1-\eta)/A$, i.e.,
$D_p \ge \tilde{c}A/(1-\eta)$, which is independent of $N$ (scale-free).
\end{proof}

From Theorem~\ref{thm:feasible-interval} and Lemma~\ref{lem:scale-free}, we conclude that in the multi-winner, non-exclusion setting, parameters $(\alpha, D_p)$ can be calibrated to satisfy O1 and O2. This feasibility holds regardless of the ordering fairness regime, as priority rents become irrelevant.
In particular, this calibration reveals a powerful, yet nuanced, scale-free property. 
When the deterrence bound is non-trivial, the required proposer deposit $D_p$ depends only on the absolute size of the colluding set, $A$, not on the total number of challengers, $N$.
This implies that the system can scale in honest participants without increasing the security deposit, provided the reward share $\alpha$ can be set below 1. 
The scale-free property is thus conditional on the colluder fraction maintaining a minimum threshold, $A/N > 1-\eta$. 
If this ratio falls below the threshold (e.g., the system becomes overwhelmingly honest), the calibration reverts to the scale-limited regime where $\alpha=1$.

\section{Discussion}
\label{sec:discussion}

In this section, we interpret our objectives and results as design trade-offs, highlight practical extensions of the baseline model, and discuss implications for ZK-based fraud proofs.

\subsection{Security objective strength and alternatives}

Our objectives \textsf{O1} and \textsf{O2} represent a deliberately strong security objective, and we briefly discuss weaker alternatives and how they change the system’s meaning.
\textsf{O1} requires ex-ante non-negative expected utility for every honest challenger in a fraudulent state, which corresponds to a permissionless setting where safety should not depend on a particular party being ``the one'' who is willing to subsidize verification.
\textsf{O2} complements this by requiring that adversarial coalitions suffer a meaningful loss, reflecting an operational constraint: even if correctness is eventually restored, repeated cheap attacks can impose external costs (congestion, delayed finality, degraded user experience) and thus are undesirable in practice.
Relaxing these contracts changes the meaning of the system.
For instance, an existential variant of \textsf{O1} (``there exists at least one honest challenger with non-negative expected utility'') shifts the protocol toward a market-dependent security model where safety becomes sensitive to participant composition and prevailing ordering/fee conditions.
Similarly, if \textsf{O2} is relaxed to a pure correctness/liveness requirement (``attacks are acceptable as long as they are eventually caught''), then attacks may still be economically cheap to repeat, and the resulting dispute overhead (fees, congestion, delays) can be borne by users or the protocol.
Our impossibility results for single-winner designs should therefore be read as statements about the tension between single-winner mechanisms and strong permissionless contracts, under realistic ordering power.

\subsection{Practical design variants and model extensions}

Our baseline model intentionally abstracts away several implementation choices to isolate the role of ordering power and winner selection. Firstly, some protocols separately reimburse challengers for the execution costs of disputing (most notably on-chain gas, and more generally verification/proving overhead), in addition to any reward funded from the proposer deposit. Modeling such reimbursement in a fully permissionless way requires a verifiable notion of incurred cost (or an enforceable cap) and protections against cost inflation or Sybil splitting, so we treat these costs as externally incurred in the baseline and view explicit cost reimbursement as a clean extension axis that mainly changes IR calibration. Secondly, while we use an equal-split rule among winners for simplicity, practical reward rules may be front-loaded or contribution-weighted to discourage free-riding by increasing rewards for early or high-information challenges. In the presence of ordering power, however, stronger early rewards can also benefit a well-prepared adversary that can reliably secure early inclusion, potentially making the mechanism more single-winner-like. Lastly, our non-exclusion results describe an idealized regime where all valid challenges can be included, whereas real deployments must choose dispute windows and capacity limits that bound how many winners can be processed. Extending the model with explicit time and capacity constraints (e.g., a cap $m$ and residual priority rents $f>0$) would quantify how finite inclusion shrinks the feasible parameter region and would provide guidance for window sizing in realistic environments.

\subsection{Implications for ZK fraud proofs and precomputation advantages}
Optimistic rollups are increasingly exploring ZK-based fraud proofs that resolve disputes via a single succinct proof submission rather than an interactive game \cite{succinct_op_succinct_lite_blog_2025}, \cite{kailua_book_intro}.
This shift can simplify disputes but also turns participation into a single-shot inclusion race, making transaction ordering and inclusion priority more consequential.
Since ZK proof generation is not instantaneous and often requires dedicated hardware and operational readiness, parties with earlier start times or standing proving infrastructure may enjoy a systematic advantage.
In particular, an adversarial proposer can anticipate the fraudulent instance \emph{before} publishing it and precompute (or partially precompute) the proving pipeline, then submit its preferred proof transaction immediately when the dispute window opens.
When dispute resolution is effectively single-winner (or supports only a very small number of includable proofs), this precomputation advantage can combine with ordering power to secure early inclusion and potentially reduce the ex-ante risk of attempting fraud.
This strengthens the case for non-exclusion approximations and for explicitly accounting for ordering power in ZK-based dispute paths.

\section{Conclusion}
\label{sec:conclusion}

Challenge-based protocols typically rely on the optimistic assumption that at least one honest validator will step up.
Our view is stricter: we study incentive mechanisms under deliberately strong security objectives, and note that weakening these objectives changes the meaning of the resulting security model.
Our analysis reveals that first-valid, single-slot designs are inherently fragile: proposers can block the inclusion of honest challenges, priority markets can drain any potential surplus, and even fair ordering mechanisms force rewards to grow unsustainably with the number of participants.
In contrast, non-exclusionary multi-winner designs address the root cause of this fragility by including every valid challenge within the dispute window, yielding clear calibration rules based on measured verification costs and anticipated collusion share.
At the same time, real deployments face time and capacity constraints, so explicit limits on includable winners can shrink the feasible region relative to the idealized non-exclusion regime.

Future work includes implementing and evaluating a non-exclusion, multi-winner design with partial burning, together with reference implementations and performance measurements for calibration.
Another direction is to stress-test the calibration with end-to-end simulations that model partially efficient markets and realistic adversaries, and to extend the model to capture the full multi-stage nature of protocols like BoLD and Dave.
Finally, as rollups explore ZK-based fraud proofs that collapse disputes into a single proof submission, it becomes increasingly important to understand how single-shot resolution and precomputation advantages interact with ordering power.

\bibliographystyle{splncs04}
\bibliography{references}

\end{document}